\documentclass[11pt,reqno]{amsproc}
\usepackage[margin=1.0in]{geometry}
\usepackage[english]{babel}
\usepackage[utf8]{inputenc}
\usepackage{ulem}
\usepackage{amsmath,amssymb,amsthm, comment, mathtools}
\usepackage{hyperref,mathrsfs,xcolor, amsfonts, amsbsy}
\usepackage{enumitem}
\usepackage{genealogytree}
\usepackage{cancel}

\newtheorem{thm}{Theorem}[section]
\newtheorem{prop}[thm]{Proposition}
\newtheorem{lem}[thm]{Lemma}
\newtheorem{cor}[thm]{Corollary}

\theoremstyle{definition}
\newtheorem{definition}[thm]{Definition}

\theoremstyle{remark}
\newtheorem{remark}[thm]{Remark}

\newcommand{\Reg}{\mathsf{Reg}}

\title{On the collapse of three point vortices on surfaces}

\author{Theodore D. Drivas}
\address{Department of Mathematics, Stony Brook University, Stony Brook, NY, 11790}
\email{tdrivas@math.stonybrook.edu}

\author{Boris A. Khanikati}
\address{Department of Physics, Harvard University, Cambridge, MA 02138, USA\newline
\indent Department of Mathematics, Stony Brook University, Stony Brook, NY, 11790}
\email{bkhanikaev@g.harvard.edu}

\author{Valeriya A. Khanikati}
\address{Department of Mathematics, University of Toronto, Toronto, M5S 2E4 ON, Canada\newline
\indent Department of Mathematics, Stony Brook University, Stony Brook, NY, 11790}
\email{valeriya.khanikaeva@stonybrook.edu}

\date{}

\begin{document}
\begin{abstract}
Point vortices represent an important reduced model describing two-dimensional ideal fluid dynamics. It is well known that there exist three-vortex configurations on the Euclidean plane $\mathbb{R}^2$ that exhibit finite-time singularities, i.e., collapse to a single point. Moreover, in $\mathbb{R}^2$, such collapses occur only self-similarly. Here, we investigate the extent to which this phenomenon persists on curved surfaces. We show that self-similar collapse is a universal feature of surfaces of nonnegative constant curvature, namely the plane and the sphere. In contrast, on the hyperbolic plane, it is shown that self-similar collapsing solutions do not exist with respect to any distance variable defined by an analytic function of the geodesic distance. Finally, we establish the existence of nearly self-similar collapse of three vortices on arbitrary smooth surfaces embedded in $\mathbb{R}^3$.
\end{abstract}
\maketitle

\section{Introduction}
The point vortex model is one of the simplest mathematical models for studying two-dimensional (2D) turbulence. By concentrating vorticity into finitely many Dirac masses of prescribed circulations, it thereby reduces the Euler evolution to a finite-dimensional system of ordinary differential equations. This reduction originates in the nineteenth century and was later recognized as a rich Hamiltonian system \cite{Helmholtz1858, Kelvin1869, Kirchhoff1876, aref2007}, where the resulting equations have widely been used as a ``mathematical playground" to study interesting and non-trivial phenomena \cite{aref2007, novikov1975, aref1983, aref1988}.  

A central appeal of the point vortex model is that it preserves several structural features of ideal hydrodynamics while being far more analytically tractable. First, as mentioned previously, the governing equations are Hamiltonian. Second, the dynamics possesses conserved quantities associated with Euclidean symmetries, which largely constrain motion at low vortex numbers. For example, for $N=3$ vortices on the plane, the system is integrable. Moreover, when expressed in trilinear coordinates for the vortex triangle, the reduced problem is governed by an integrable two-dimensional system \cite{Synge1949}. Given how amenable the three-vortex problem is to detailed analysis, a thorough classification of various solutions has been done, including equilibria, periodic motion, and singular solutions \cite{Groebli1877}. 

Among these singular solutions are those of \textit{finite-time collapse}: configurations in which vortices coalesce to a single point in finite time. This type of phenomenon shall be the focus of this paper. 

The collapse of three point vortices on the plane is a classical and well-understood phenomenon \cite{leoncini2000collapse, Krishnamurthy, aref2010}. Owing to the conservation laws of the system, the necessary and sufficient conditions for collapse imply that the motion is necessarily self-similar: the triangle formed by the vortices remains similar to itself throughout the evolution. 

It is perhaps natural to extend the motion of point vortices to surfaces beyond the plane; one such motivation is that the Earth may be idealized as a sphere supporting a nearly two-dimensional fluid atmosphere. The behavior of point vortices on surfaces is relatively well studied, particularly for closed surfaces \cite{kimura1999vortexConst, Kimura1987VortexMO, drivas2024vortexpairs, Boatto2008VorticesOC, dritschel2015vortices, bjorn2022, sakajo2016}. In fact, on the sphere specifically, the existence of collapsing configurations has been known for quite some time now \cite{kidambi1998sphere}. 

Kidambi and Newton studied self-similar collapse on the sphere and derived necessary and sufficient conditions for its occurrence. Since the relevant first integrals on the sphere depend on the mutual chord lengths in exactly the same way as in the planar problem, we here show that the classical planar argument extends directly to the spherical setting. As a consequence, every collapsing three-vortex configuration on the sphere collapses self-similarly with respect to chord length.

It is then interesting to consider on which surfaces specifically collapse persists self-similarly. The existence of curvature essentially introduces some length-scale, and given that curvature is not constant, such collapse is quite unlikely. A natural hypothesis would be that self-similar collapse exists on surfaces of constant curvature, e.g., the plane, sphere, or hyperbolic plane. We show that such an assumption is not true, and the answer depends on the sign of the curvature: while collapse on the plane and sphere is necessarily self-similar, the hyperbolic plane admits no self-similar collapsing solutions. This result is summarized in the following theorem. 

\begin{thm}\label{thm:1-1}
For the three-vortex problem, every collapsing configuration on the plane or sphere is self-similar, where self-similarity is understood with respect to Euclidean distance in the planar case and chord length in the spherical case. By contrast, the hyperbolic plane admits no self-similar collapsing solutions with respect to any distance variable given by an analytic function of the geodesic distance.
\end{thm}

 We remark that Vil\`{a}s \cite{simo2022n} proved a similar result for the real projective plane, namely that there is no self similar collapse (in terms of the Euclidean chord distance inherited from the standard representation of $\mathbb R\mathbb P^2$ as $S^2/\mathbb Z_2$) of three vortices in $\mathbb{R}\mathbb{P}^2$.
This result, along with ours for the hyperbolic plane, disentangles the issue of integrability of vortex motion (which holds in all cases for three vortices \cite{modin2020integrability}) and self-similar collapse.

This, however, does not rule out collapse on surfaces in general, i.e., it is possible that solutions that collapse in finite time do exist on arbitrary surfaces. Recent work has shown that self-similar collapses (and bursts) are robust with respect to sufficiently regular perturbations \cite{grotto2022burst}. Considering three point vortices on a surface locally, the motion is essentially planar with higher order corrections, and with the aforementioned work in hand, our second main result is the existence of collapsing configurations on a closed surface $S$. 

Even though collapse will not, in general, be self-similar, the motion will asymptotically approach self-similarity. With this motivation, the work of Vil\`{a}s \cite{simo2022n} introduced the notion of asymptotically (or, near) self-similar collapse:
\begin{definition}[Asymptotically self-similar collapse]
We say that a three point vortex system undergoes asymptotically self-similar collapse in terms of a variable $\ell_{ij}$ if all sides decay at the same rate as the time approaches the time of collapse, $t_c$. That is,
\begin{equation}
\frac{\dot{\ell}_{ij}/\ell_{ij}}{\dot{\ell}_{km}/\ell_{km}} \rightarrow 1, \qquad \text{as}~ t\to t_c. 
\end{equation}
\label{def.near-self-similar}
\end{definition}
Hence, our second main result is summarized in the following theorem. 

\begin{thm}\label{prop:main}
Let $S$ be a surface embedded in $\mathbb R^3$ and fix $p\in S$. Let $(\Gamma_1, \Gamma_2, \Gamma_3)$ admit a planar self-similar collapse. Then, there exists $\varepsilon_0>0$ such that for every $0 < \varepsilon < \varepsilon_0$ there are initial positions $z_i(0) \in S$ with $d_S(z_i(0), p) < \varepsilon$ for $i = 1, 2, 3$ so that the solution of the point vortex equations on $S$ with these initial data collapses to $p$ in an asymptotically self-similar fashion. 
\end{thm}

\begin{remark}
As will be seen from the proof of the theorem below, the existence of collapse hinges on local data. As such, it is sufficient to assume that $S$ is a Riemannian surface, but for the sake of concreteness, we treat embedded surfaces.  Moreover, if $S$ is a surface of higher genus ($\mathfrak g\geq 1)$, collapsing configurations will still exist. For $\mathfrak g\geq 1$, the equations of motion must be modified to accommodate a non-trivial irrotational part of the velocity \cite{GustafssonB}. However, the procedure of proof of the proposition remains the same, as for sufficiently small neighborhoods on $S$, the irrotational part will be perturbative and  the  conclusion holds. Collapse on surfaces with boundary follow similarly. 
\end{remark}

Let us now review briefly the point vortex model, collapsing solutions in the Euclidean plane, sphere and hyperbolic plane, and point vortices on surfaces.


\section{Point vortex model}

\subsection{Point vortices on the plane}
Point vortices are a description of 2D turbulence by representing the vorticity profile as being concentrated in finitely many points:
\begin{equation}
     \omega(x,t) = \sum_{i}\Gamma_i  \delta_{z_i(t)}(x),
\end{equation} where $z_i \in \mathbb{R}^2,$  denotes the $i^{th}$ vortex with corresponding circulation $\Gamma_i\in \mathbb R.$  As mentioned previously, the point vortex dynamics form a Hamiltonian system with Hamiltonian function
\begin{equation}\label{eq:hamiltonian}
    H = - \frac{1}{4\pi} \sum_{i\neq j} \Gamma_i \Gamma_j \log|z_i - z_j|.
\end{equation}
Here, $|\cdot|$ is the Euclidean distance in $\mathbb R^2$. The system also has two other first integrals of motion, the linear and angular momenta:
\begin{align*}
    P &\equiv \sum_i \Gamma_i z_i(t) = \mathrm{const.}, \\
    J &\equiv \sum_i \Gamma_i z_i^2(t) = \mathrm{const.}
\end{align*}
For three vortices, a  useful conserved quantity $L$ depending only on relative distances may be formed from a linear combination of $P$ and $J$:
\begin{equation*}
    L = \sum_{i<j} \Gamma_i  \Gamma_j \,\ell_{ij}^2,
\end{equation*}
where $\ell_{ij} = |z_i - z_j|$ is the Euclidean distance between the vortices. 
Three point vortex collapse  can be completely characterized in terms of this invariant $L$  \cite{Krishnamurthy}. Namely,
\begin{lem}\label{thm:plane-collapse}
The necessary and sufficient conditions for collapse are that $L = 0$ and
\begin{equation*}
    \Gamma_1 \Gamma_2 + \Gamma_2 \Gamma_3 + \Gamma_3 \Gamma_1 = 0.
\end{equation*}
Furthermore, these conditions force the collapse to be self-similar. 
\end{lem}
\begin{proof} We include the proof, which follows \cite{Krishnamurthy}, for the sake of completeness.
Let us rewrite the point vortex triangle in terms of the radius of the circumcircle, $R(t)$, and the interior angles $\alpha + \beta + \gamma = \pi$. We may write 
\begin{equation*}
    L = R^2 \left(\frac{\sin^2\alpha}{\Gamma_1} + \frac{\sin^2 \beta}{\Gamma_2} + \frac{\sin^2 \gamma}{\Gamma_3} \right). 
\end{equation*}
Since we assume that collapse occurs at some finite time, there exists a time $t_c$ such that $R(t_c) = 0$. It then follows that $L(t_c) = 0$, and by its invariance, $L = 0$. Hence, 
\begin{equation}\label{eq:cond1}
    \frac{\sin^2\alpha}{\Gamma_1} + \frac{\sin^2 \beta}{\Gamma_2} + \frac{\sin^2 \gamma}{\Gamma_3} = 0.
\end{equation}
Let us call $M = \Gamma_1\Gamma_2 + \Gamma_2 \Gamma_3 + \Gamma_3 \Gamma_1$. We may rewrite the Hamiltonian as
\begin{equation*}
    e^{-2\pi H/\Gamma_1 \Gamma_2 \Gamma_3} = \left(\frac{R(t)}{R(0)}\right)^{M/\Gamma_1\Gamma_2\Gamma_3} (\sin\alpha)^{1/\Gamma_1}(\sin\beta)^{1/\Gamma_2} (\sin \gamma)^{1/\Gamma_3}.
\end{equation*}
Assume for the sake of contradiction that $M \neq 0$. As $t\to t_c$, $R(t) \to 0$. Hence, for the quantity to remain constant and positive, it would require at least one angle to approach $0$ or $\pi$, forcing the triangle to become collinear. However, in this case, $R \to \infty$, a contradiction. Hence, $M = 0$.

With these conditions in hand, we may rewrite 
\begin{equation}\label{eq:cond2}
    H = - \frac{\Gamma_1 \Gamma_2\Gamma_3}{2\pi}\left(\frac{\log \sin \alpha}{\Gamma_1} + \frac{\log \sin \beta}{\Gamma_2} + \frac{\log \sin \gamma}{\Gamma_3} \right).
\end{equation}
Now, we have three equations for the three angles: they sum to $\pi$,  (\ref{eq:cond1}), and (\ref{eq:cond2}). Three independent equations for three variables, hence we get that $\alpha, \beta, \gamma$ are constant. Hence, the motion is self-similar; i.e. there exists a common scaling $\lambda(t)$ such that 
\begin{equation}
    \ell_{ij}(t) = \lambda(t) \ell_{ij}(0).
\end{equation}
At $t=0$, of course, $\lambda(0) =1$. For finite-time collapse to occur, there must exist a finite time $t_c$ where $\lambda(t_c) = 0$.   This is observed to occur when the stated conditions are met.
\end{proof}

\begin{remark}
    We note that the collapse of two vortices cannot occur: this is easy to see from the conserved quantity $L.$ Particularly, if collapse occurs, then $L = \Gamma_1 \Gamma_2 \ell _{12}(t)^2 = 0$ for all $t$. However, this would imply that at least one of the circulations must be zero.
\end{remark}

\subsection{Motion of point vortices on surfaces}
Here, we generalize the motion in the above subsection to surfaces. The Hamiltonian of the point vortex system is given in terms of the Green's function of the Laplace-Beltrami operator, 
\begin{equation}
    H(z_1, \ldots, z_n) = \frac{1}{2}\sum_{i \neq j}\Gamma_i \Gamma_j G_S(z_i, z_j) + \frac{1}{2}\sum_{\ell} \Gamma_\ell^2 R_S(z_\ell, z_\ell),
\end{equation}
where $R_S(\cdot, \cdot)$ represents the Robin function, the regular part of the Green's function:
\begin{equation}
    R_S(p, q) = 
    \begin{cases}
        G_S(p, q) + \log (d_S(p, q)/2\pi), \quad &p \neq q  \\
        \lim_{p \to q} [G_S(p, q) + \log (d_S(p, q)/2\pi)], &p = q.
    \end{cases}
\end{equation}
Here, $d_S(\cdot, \cdot)$ is the geodesic distance on surface $S$. Let us introduce the notation that if a function depends on all vortex coordinates, we shall write $f(z_1, \ldots, z_n) \equiv f(z)$. The equations of motion corresponding to our Hamiltonian are given by
\begin{equation}
    \Gamma_i\dot z_i = J(z_i)\widetilde\nabla_{z_i} H(z),
\end{equation}
where $\widetilde \nabla$ is the covariant derivative associated with the induced metric, and $J$ represents the almost complex structure on $S$ (rotation by ninety degrees in the tangent plane).

The above Hamiltonian is true for closed surfaces of genus $\mathfrak g = 0$ (so that there are no non-trivial harmonic vector fields). For surfaces of higher genus or non-compact surfaces, modifications must be made. For example, recent work has derived the equations of motion for point vortices on a torus \cite{sakajo2016}, and the recent work \cite{grotta2024interplay} gives a comprehensive review on how the vortex system couples to evolving potential flow. Our proof of the existence of collapsing configurations deals with surfaces of genus zero, but can be naturally extended to surfaces for $\mathfrak g \geq 1$ since it is, in essence, perturbative.  Similarly, on surfaces with boundary, the effect of the Robin function is perturbative in the interior, see e.g. \cite{drivas2025pensive}, and the collapse will again follow.

Now that we have reviewed some of the existing theory for point vortices, we proceed to our discussion on the self-similar motion of point vortices on surfaces of constant curvature.

\section{Surfaces of constant curvature}
We consider here surfaces of constant curvature: (i) of constant positive curvature, the unit sphere $\mathbb{S}^2$, and (ii) constant negative curvature, the hyperbolic plane $\mathbb{H}^2$. The Green's functions on these surfaces are given by \cite{kimura1999vortexConst}
\begin{align}\label{eq:green-sphere}
    G_{\mathbb{S}^2}(p, q) &= - \frac{1}{2\pi}\log \left(\sin \frac{d_{\mathbb{S}^2}(p, q)}{2} \right), \quad \mathrm{for} ~p, q \in \mathbb{S}^2, \\
    G_{\mathbb{H}^2}(p, q) &= - \frac{1}{2\pi}\log \left(\tanh \frac{d_{\mathbb{H}^2}(p, q)}{2} \right), \quad \mathrm{for} ~p, q \in \mathbb{H}^2.
\end{align}
One may notice that the Green's functions on $\mathbb S^2$ and $\mathbb H^2$ are not directly analogous. Unlike the hyperbolic plane, $\mathbb S^2$ is a compact manifold without boundary, and hence there exists no distribution satisfying $\Delta_p G(p, q)= \delta_q(p)$. Rather, one defines the normalized Green's function by $\Delta_p G(p, q) = \delta_q(p) - 1/\mathrm{Area}(\mathbb S^2)$. As a result, one may add a uniform vorticity background to the Green's function without affecting the equations. Indeed, writing $G(d_{\mathbb S^2}) = -\frac{1}{2\pi}\log(\tan d_{\mathbb S^2}/2)$ differs only by such a background, and hence reproduces the same equation. On the other hand, $G(d_{\mathbb H^2}) = -\frac{1}{2\pi}\log(\tanh d_{\mathbb H^2}/2)$ is the unique choice on $\mathbb H^2$, since any (radially symmetric) harmonic addition must satisfy the prescribed behavior at infinity.

\subsection{The sphere}
Given that the explicit form of the Green's function is known, it is possible to simplify the form of the Hamiltonian. First, we may calculate the Robin function as 
\begin{equation*}
    R_{\mathbb{S}^2}(p, q) = - \frac{1}{2\pi}\log \left(\frac{\sin(d_{\mathbb{S}^2}(p, q)/2)}{d_{\mathbb{S}^2}(p, q)} \right).
\end{equation*}
The diagonal term is evaluated as the limit
\begin{equation*}
    R_{\mathbb{S}^2}(q, q) = \lim_{p \to q} R_{\mathbb{S}^2}(p, q) = \frac{1}{2\pi}\log 2.
\end{equation*}
Hence, we arrive at
\begin{equation}
    H_{\mathbb{S}^2}(z) = - \frac{1}{2\pi}\sum_{i<j} \Gamma_i \Gamma_j \log\left(\sin \frac{d_S(z_i, z_j)}{2}\right) + \frac{\log 2}{4\pi}\sum_{\ell = 1}^N \Gamma_\ell^2. 
\end{equation}
Note that the argument of the logarithm in the above equation is proportional to the chord length between two points on the sphere, $\ell_{ij}$, with which we rewrite 
\begin{equation}
    H_{\mathbb{S}^2}(z) = -\frac{1}{2\pi} \sum_{i<j}\Gamma_j \Gamma_j \log \ell_{ij} + \mathrm{const.},
\end{equation}
where the constant does not affect the dynamics. Correspondingly, the equations of motion are 
\begin{equation}
    \dot z_i = \frac{1}{2\pi}\sum_{j \neq i} \Gamma_j \frac{z_i \times z_j}{|z_i - z_j|^2} = \frac{1}{4\pi} \sum_{j \neq i} \Gamma_j \frac{z_i\times z_j}{1 - z_i \cdot z_j},
    \label{eq: eom for sphere}
\end{equation}
which reproduces the equations given in \cite{kidambi1998sphere}. We may also write equations for length,
\begin{equation}
    \frac{d\ell_{ij}^2}{dt} = \frac{\Gamma_k V}{\pi R}\left(\frac{1}{\ell_{jk}^2} - \frac{1}{\ell_{ki}^2} \right),
\end{equation}
where $V$ is the volume of the parallelepiped formed by the vectors.  

The motion of point vortices on a sphere was studied extensively by Kidambi and Newton \cite{kidambi1998sphere}, including the case of three point vortex collapse. In particular, they prove the following theorem with regard to the necessary and sufficient conditions for self-similar collapse on the sphere:
\begin{lem}[Kidambi and Newton \cite{kidambi1998sphere}]
The necessary and sufficient conditions for self-similar collapse on a sphere in terms of the chord length $\ell_{ij}$ are 
\begin{align}\label{eq:L-sphere}
    L = \sum_{i<j} \Gamma_i \Gamma_j \ell_{ij}^2 = 0 \qquad  \text{and} \qquad
    \sum_i \frac{1}{\Gamma_i} = 0, 
\end{align}
and that the vortices do not form an equilibrium. 
\end{lem}

While it is shown that these conditions are necessary and sufficient, it is also possible to demonstrate that collapse always occurs self-similarly, by analogy with the planar motion. It will be useful to note that the quantity $L$ given in (\ref{eq:L-sphere}) is conserved. We sketch the proof here.

\begin{proof}
The planar proof of Lemma \ref{thm:plane-collapse} relies only on the conserved quantities $L$ and $H$. On the sphere, these quantities have exactly the same dependence on the mutual distances when expressed in terms of chord lengths $\ell_{ij}$. Therefore, the planar argument applies verbatim: if collapse occurs, then \(L=0\), and conservation of the Hamiltonian forces the ratios of the three chord lengths to remain constant. Hence, all three chord lengths shrink by a single common scale factor, so the collapse is self-similar.
\end{proof}

We next remark that this collapse, while it is self-similar in terms of chord lengths (an extrinsic concept), it is not so in terms of the intrinsic geodesic length:
\begin{cor}
Collapse on the sphere is not self-similar in terms of the geodesic length. 
\end{cor}
\begin{proof}
Given collapse on $\mathbb{S}^2$, there exists $\lambda(t)$ such that $\ell_{ij}(t) = \lambda(t) \ell_{ij}(0)$. Using the relation between geodesic distance and chord length $d_{ij} = 2 \arcsin(\ell_{ij}/2)$, 
\begin{equation*}
    d_{ij}(t) = 2 \arcsin \left(\frac{\ell_{ij}(0)}{2}\lambda(t) \right).
\end{equation*}
However, $\arcsin(\lambda x) \neq \mu(\lambda)  \arcsin(x)$,
hence there does not exist a function $\mu(t)$ such that $d_{ij}(t) = \mu(t) d_{ij}(0)$.
\end{proof} 

\subsection{The hyperbolic plane}
We now switch to the surface of constant negative curvature, the hyperbolic plane. Following the same procedure as with the sphere, we find
\begin{equation*}
    R_{\mathbb{H}^2}(p, q) = - \frac{1}{2\pi}\log\left(\frac{\tanh(d_{\mathbb{H}^2}(p, q)/2)}{d_{\mathbb{H}^2}(p, q)} \right), \quad \text{and} \quad     R_{\mathbb{H}^2}(q, q) = \frac{\log 2}{2\pi}.
\end{equation*}
Thus, the Hamiltonian for the vortex system is
\begin{equation}
\label{HamiltonianH2}
    H_{\mathbb{H}^2}(z) = - \frac{1}{2\pi} \sum_{i<j} \Gamma_i \Gamma_j \log\left(\tanh \frac{d_{\mathbb{H}^2}(z_i, z_j)}{2} \right) + \frac{\log 2}{4\pi}\sum_{\ell = 1}^N \Gamma_{\ell}^2.
\end{equation}
Henceforth, let us ignore the constant term in the Hamiltonian. Here, the argument of the logarithm is known as the \textit{pseudohyperbolic distance}, and we shall denote it as
\begin{equation}
    \rho_{ij}\equiv  \tanh\left(\frac{d_{\mathbb{H}^2}(z_i, z_j)}{2} \right).
\end{equation}
\begin{remark}[Potential flows on the hyperbolic plane]
    It is important to note a subtlety regarding the point vortex system defined by (\ref{HamiltonianH2}). In the space of $L^2$ velocity fields, the velocity is not uniquely determined by the vorticity. Rather, one may decompose the velocity field into a vorticity-induced part $K[\omega]$ and a harmonic part, that is, $u = K[\omega] + u_h$, where $u_h$ is an $L^2$-harmonic vector field. On $\mathbb{H}^2$, the space of such harmonic fields is infinite-dimensional \cite{khesin2012hyperbolic}. Therefore, the point vortex system taken above retains only the vorticity-induced component $K[\omega]$ of the velocity field, corresponding to the choice $u_h = 0$. However, it is not automatic from this construction that an initial choice $u_h(0)=0$ remains zero throughout the Euler dynamics, see discussion in \cite{drivas2023singularity}. As such, flows with nontrivial harmonic component are not necessarily represented by this point vortex model. A more honest reduced model for highly concentrated vortices should follow from considerations similar to those in \cite{grotta2024interplay}.
\end{remark}

In principle, one could define self-similar collapse with respect to any two-point function $f(\cdot, \cdot)$. However, the structure of the Hamiltonian singles out a natural choice: the quantity appearing inside the logarithm. In the hyperbolic plane $\mathbb{H}^2$, this is (up to a linear rescaling) the pseudohyperbolic distance $\rho_{ij}$. It is therefore most natural to define self-similarity in terms of said variable. Indeed, this choice is consistent with the Euclidean and spherical cases, where the Hamiltonian depends logarithmically on the chord length, and self-similar collapse correspondingly occurs in terms of that quantity.

We would like to make use of the conserved quantities. We already have the Hamiltonian. Furthermore, following \cite{montaldi2014hyperbolic}, there exists a momentum map $\mathbf{J} = \sum_i \Gamma_i X_i$ which is conserved. The Hyperbolic plane is endowed with the Lorentz inner product 
\begin{equation*}
    \langle x, y \rangle = x_1 y_1 + x_2 y_2 - x_3 y_3.
\end{equation*}
The Lorentz norm of the momentum map $\mathbf J$ must be conserved as well:
\begin{equation*}
    \langle \mathbf J,\mathbf J\rangle_L = \sum_i \Gamma_i^2\langle X_i,X_i\rangle_L + 2\sum_{i<j}\Gamma_i\Gamma_j\langle X_i,X_j\rangle_L.
\end{equation*}
Since $X_i \in \mathbb{H}^2$, one has $\langle X_i,X_i\rangle_L=-1$, while
\begin{equation*}
    \langle X_i,X_j\rangle_L=-\cosh d_{ij} = -\frac{1+\rho_{ij}^2}{1-\rho_{ij}^2},.
\end{equation*}
Therefore
\begin{equation*}
    \langle \mathbf J,\mathbf J\rangle_L  = -\sum_i \Gamma_i^2 - 2\sum_{i<j}\Gamma_i\Gamma_j\frac{1+\rho_{ij}^2}{1-\rho_{ij}^2} = -\left(\sum_i \Gamma_i\right)^2 - 4\sum_{i<j}\Gamma_i\Gamma_j\frac{\rho_{ij}^2}{1-\rho_{ij}^2}.
\end{equation*}
Since $\langle \mathbf J,\mathbf J\rangle_L$ and $\sum_i\Gamma_i$ are conserved, it follows that
\begin{equation}
    L \equiv \sum_{i<j}\Gamma_i\Gamma_j\frac{\rho_{ij}^2}{1-\rho_{ij}^2}
\end{equation}
is also conserved.  With two conserved quantities now known, we may write the necessary conditions for self-similar collapse on $\mathbb{H}^2$ in terms of the pseudohyperbolic distance. 
\begin{lem}\label{lem:necessary-H2}
    If a three point vortex configuration collapses self-similarly in terms of $\rho_{ij}$, the necessary conditions for collapse on $\mathbb{H}^2$ are that (i): 
\begin{equation}
    \sum_{i<j} \Gamma_i \Gamma_j = 0,
\end{equation}
and (ii):
\begin{equation}
    \sum_{i<j} \Gamma_i \Gamma_j \frac{\rho_{ij}^2}{1 - \rho_{ij}^2} = 0.
\end{equation}
\end{lem}
\begin{proof}
First, consider the conserved quantity $L$. Since $\rho_{ij}$ must go to zero for collapse to occur, and given that it is conserved, $L = 0$ for all times.
For conservation of the Hamiltonian, it must be true that
\begin{equation*}
    \sum \Gamma_i \Gamma_j \log \rho_{ij}(t) = \mathrm{const}.
\end{equation*}
If collapse is self-similar, there exists $\lambda(t)$ such that $\rho_{ij}(t) = \lambda(t)\rho_{ij}(0)$. Substituting this into $H$, 
\begin{equation*}
    H = - \frac{1}{2\pi}\log \lambda(t)\sum_{i<j}\Gamma_i \Gamma_j +\mathrm{const}.
\end{equation*}
The Hamiltonian is an integral of motion, and since $\lambda(t)$ is nonconstant, the coefficient in front of $\lambda(t)$ must vanish. We conclude that 
\begin{equation*}
    \sum_{i<j} \Gamma_i \Gamma_j = 0.
\end{equation*}
\end{proof}

However, the conditions become too restrictive. 
\begin{prop}\label{prop:self-similar-pseudo}
There does not exist self-similar collapse in terms of the pseudohyperbolic distance on $\mathbb{H}^2$. 
\end{prop}
\begin{proof}
By Lemma \ref{lem:necessary-H2}, any solution collapsing self-similarly must satisfy
\begin{equation*}
    \sum_{i<j} \Gamma_i \Gamma_j \frac{\rho_{ij}^2}{1-\rho_{ij}^2} = 0,
\end{equation*}
which is conserved. Suppose, for the sake of contradiction, that a self-similar collapsing solution exists. Then there is a function $\lambda(t)$ such that
\begin{equation*}
    \rho_{ij}(t)=\lambda(t)\rho_{ij}(0).
\end{equation*}
Substituting into the conserved quantity gives
\begin{equation*}
    \sum_{i<j} \Gamma_i \Gamma_j  \frac{\lambda^2(t)\rho_{ij}^2(0)}{1-\lambda^2(t)\rho_{ij}^2(0)} = 0.
\end{equation*}
We expand $\frac{\lambda^2 x}{1-\lambda^2 x}  =  \sum_{k=1}^\infty \lambda^{2k} x^k$ to obtain
\begin{equation*}
    \sum_{k=1}^\infty \lambda^{2k}(t) \sum_{i<j} \Gamma_i \Gamma_j [\rho_{ij}(0)]^{2k} = 0.
\end{equation*}
Given that $\lambda(t)$ is nonconstant, which is must be for collapse to occur, it follows that each coefficient must vanish:
\begin{equation*}
    \sum_{i<j} \Gamma_i \Gamma_j [\rho_{ij}(0)]^{2k} = 0, \quad \mathrm{for~all~} k\in\mathbb{N}.
\end{equation*}

Take the first three equations corresponding to $k = 1, 2,3$, that is, $\sum \Gamma_i \Gamma_j [\rho_{ij}(0)]^2 = 0$, $\sum \Gamma_i \Gamma_j [\rho_{ij}(0)]^4 = 0$, and $\sum \Gamma_i \Gamma_j [\rho_{ij}(0)]^6 = 0$. The equations may be written in matrix form, 
\begin{equation*}
    \begin{pmatrix}
        \rho_{12}(0)^2 & \rho_{13}^2(0) & \rho_{23}^2(0)\\
        \rho_{12}(0)^4 & \rho_{13}^4(0) & \rho_{23}^4(0)\\
        \rho_{12}(0)^6 & \rho_{13}^6(0) & \rho_{23}^6(0)
    \end{pmatrix}
    \begin{pmatrix}
        \Gamma_1 \Gamma_2 \\ \Gamma_1 \Gamma_3 \\ \Gamma_2 \Gamma_3
    \end{pmatrix} = 
    \begin{pmatrix}
        0 \\ 0 \\ 0
    \end{pmatrix}.
\end{equation*}
Assume that $\rho_{ij}$ are pairwise distinct for all pairs $i, j$. The determinant of the coefficient matrix is nonzero, hence such a system only has a solution if $\Gamma_i \Gamma_j = 0$ for all pairs $i, j$, which is impossible since the circulations are nonzero. 

Hence, suppose that two squared distances are equal. Then, we obtain $\Gamma_i \Gamma_j + \Gamma_i \Gamma_k = 0$ and $\Gamma_j \Gamma_k = 0$. Again, since the circulations are nonzero, this is impossible. 

As such, the only remaining possibility is if all initial lengths are equal, that is, $\rho_{ij} = \rho_{km}$. However, such a configuration is a relative equilibrium \cite{montaldi2014hyperbolic}, which contradicts the assumption that collapse occurs. Therefore, self-similar collapse does not occur on $\mathbb{H}^2$ in terms of the pseudohyperbolic distance.
\end{proof}

Thus far, we have considered self-similar collapse with respect to the variable analogous to the planar and spherical cases. Perhaps, it could be that self-similar collapse exists in some different variable that is a function of the geodesic (or equivalently, pseudohyperbolic) distance. Let us consider a more general case. 

In particular, we shall defined distance through analytic functions of $\rho_{ij}$, $f(\rho)$, and rule out self-similar collapse with respect to this distance. Note that the pseudohyperbolic distance is related to the geodesic distance by an analytic function on $\mathbb{H}^2$, hence $f(\rho(d))$ is itself an analytic function, so the following proposition also holds for distance functions defined in terms of analytic functions of the geodesic distance $g(d_{ij})$.

\begin{prop}
    Let $f(x)$ be a strictly increasing analytic function such that $f(0) = 0$. Then, there exists no self-similar collapse in terms of $r_{ij} = f(\rho_{ij})$ distance on $\mathbb{H}^2$.
\end{prop}

\begin{proof}
Since $f$ is analytic and $f(0) = 0$, it admits a Taylor expansion of the form
\begin{equation}
    f(x) = a_m x^m + a_{m+1}x^{m+1} + \cdots,
\end{equation}
where $m \geq 1$ is the order of zero of $f$ at zero. Since $f$ is strictly increasing, the first nonzero coefficient is positive $a_m > 0$. Let us define the function $F(x) = [f(x)]^{1/m} = x(a_m + a_{m+1}x + \cdots)^{1/m}$. Indeed, $F$ is itself analytic at zero, $F(0) = 0$, and $F'(0) = a_m^{1/m} \neq 0$. Hence, by the inverse function theorem, $h = F^{-1}$ is analytic in a neighborhood of zero. Let us define $R_{ij} = F(\rho_{ij})$. 

Suppose there exists self-similar collapse with respect to $f$, that is, $r_{ij}(t) = \lambda(t) r_{ij}(0)$. Then, $R_{ij}(t) = [\lambda(t)]^{1/m} R_{ij}(0)$, hence self-similarity in $r = f(\rho)$ implies self-similarity in terms of $R = F(\rho)$. It is therefore sufficient to show that there exists no self-similar collapse in terms of $F(\rho_{ij})$. 

By definition, $h$ has a simple zero at $0$. Since it is analytic, we may Taylor expand it and find 
\begin{equation*}
    \frac{h'(x)}{h(x)} = \frac{1}{x}+\sum_{n = 0}^\infty c_n x^n.
\end{equation*}
Conservation of the Hamiltonian $dH/dt = 0$ yields the condition 
\begin{equation*}
    0 = - \frac{1}{2\pi}\sum_{i<j} \Gamma_i \Gamma_j R_{ij}(0)\frac{h'(\mu(t) R_{ij}(0))}{h(\mu(t) R_{ij}(0))}
\end{equation*}
where we have denoted $\mu(t) = [\lambda(t)]^{1/m}$. Substituting the expansion yields 
\begin{equation*}
    0 = \frac{1}{\mu}\sum_{i<j}\Gamma_i \Gamma_j + \sum_{n \geq 0} c_n \mu^n \sum_{i<j} \Gamma_i \Gamma_j [R_{ij}(0)]^{n+1}.
\end{equation*}
Since this holds for an interval of values of $\mu$, all coefficients must vanish: 
$
    \sum_{i<j}\Gamma_i \Gamma_j = 0
$
and 
\begin{equation*}
    \sum_{i<j} \Gamma_i \Gamma_j [R_{ij}(0)]^{n+1} = 0\quad \mathrm{for ~ all} ~ n~ \mathrm{with~} c_n \neq 0.
\end{equation*}
Since $L \to 0$ at collapse, it follows that $L = 0$ for all times. Recall that $h$ has a simple zero at $0$. Then, let us write 
\begin{equation*}
    h^2/(1 - h^2) = \sum_{n\geq 2} b_n x^n.
\end{equation*} 
Indeed, since $h$ has a simple zero at the origin, $b_2 > 0$. Moreover, because $L = 0$, we also have 
\begin{equation*}
    0 = \sum_{n=2}^\infty b_n \mu^n \sum_{i<j} \Gamma_i \Gamma_j [R_{ij}(0)]^n. 
\end{equation*}
Thus, whenever $b_n \neq 0$, this implies that $\sum_{i<j} \Gamma_i \Gamma_j [R_{ij}(0)]^n = 0$. 

To draw a contradiction similar to that in Prop. \ref{prop:self-similar-pseudo}, we wish to show that we are guaranteed an equation of the form $\sum \Gamma_i \Gamma_j [R_{ij}(0)]^q = 0$ for some $q\neq 0, 2$. If $b_q \neq 0$ for some $q \neq 2$, this is immediate. Otherwise, 
\begin{equation*}
    \frac{h^2(x)}{1 - h^2(x)} = b_2 x^2. 
\end{equation*}
Since $b_2 > 0$,
\begin{equation*}
    h(x) = \frac{\sqrt{b_2}x}{\sqrt{1 + b_2 x^2}}. 
\end{equation*}
Consequently, 
\begin{equation*}
    \frac{h'(x)}{h(x)} = \frac{1}{x} - b_2 x + b_2^2 x^3 - \cdots.
\end{equation*}
Thus, let $c_3 = b_2^2 \neq0$, and the Hamiltonian supplies an equation with exponent $q = 4$. 

We therefore arrive at three conditions: 
\begin{equation*}
   (i): \sum_{i<j} \Gamma_i\Gamma_j = 0, \quad (ii): \sum_{i<j} \Gamma_i \Gamma_j [R_{ij}(0)]^2 = 0, \quad (iii): \sum_{i<j} [R_{ij}(0)]^q = 0
\end{equation*}
for some $q \neq 0, 2$. Now, an analogous argument to that of Prof. \ref{prop:self-similar-pseudo} follows. If $R_{ij}(0)$ are pairwise distinct for all pairs $i, j$, then the system is invertible and it follows that $\Gamma_i \Gamma_j = 0$ for all pairs $i, j$, leading to trivial circulations. The system only has a non-trivial solution if the initial configuration is equilateral; however, such a configuration on $\mathbb H^2$ is a relative equilibrium \cite{montaldi2014hyperbolic}, thus collapse does not occur. 
\end{proof}

The absence of such self-similar collapses in $\mathbb H^2$ highlights a fundamental difference between negative and nonnegative curvature. Let us try to interpret this difference geometrically. On surfaces, the separation of nearby geodesics (the geodesic deviation) satisfies the Jacobi equation. For a surface of constant curvature $K$ (we take $K = \{\pm 1, 0\}$ for simplicity), we may write it as $S''(r) + K\,S(r) = 0$. The solution to this equation is 
\begin{equation*}
    S(r) = \begin{cases}
        \sin r, \quad & K = 1, \\
        r, & K = 0,\\ 
        \sinh( r), \quad & K < 0. 
    \end{cases}
\end{equation*}
The function $S(r)$ defines the angular component of the metric. Equivalently, $2\pi S(r)$ is the circumference of a geodesic circle of radius $r$, so $S(r) dr$ may be thought of the ``density" of geodesic circles. 

The collapse of three point vortices hinges on two conserved quantities: the Hamiltonian and $L$, which arise from the Green's function and momentum map, respectively. Indeed, both the Green's function and momentum map probe this same density in two complementary ways. The momentum map, which is associated with rotations, may be thought of as measuring circulation weighted by the enclosed area, obtained by integrating this density. Consequently, the pairwise momentum invariant $L$ naturally involves $\sin^2(d_s/2)$ and $\sinh^2(d_s/2)$ for the sphere and hyperbolic plane, respectively. The Green's function is the potential generated by a point source. Conservation of flux implies that its radial derivative is inversely proportional to the circumference of a geodesic circle, so the Green's functions naturally involve $\tan(d_s/2)$ and $\tanh(d_s/2)$. 

The absence of self-similar collapses above suggests that self-similarity is compatible with the conservation of $L$ and $H$ only when they depend on the same distance variable, or on variables related by sufficiently simple functions. In that case, the scale factor ``separates", and conservation reduces to conditions on the initial data. On the contrary, if the variables do not agree, requiring conservation generates additional shape-dependent constraints, which may form general obstructions. 

As discussed previously, the sphere admits a uniform vorticity background, allowing one to ``match" the form of the Green's function to the variable in $L$, namely, equation (\ref{eq:green-sphere}). On the other hand, this is not possible on surfaces of negative curvature, owing to the exponential growth of volume on $\mathbb H^2$ and the prescribed behaviour of the Green's function at infinity.

\begin{remark}[Collapse in terms of extrinsic distances]
On the sphere, the collapse is self-similar in terms of an extrinsic distance (Euclidean chord distance), rather than an intrinsic one (geodesic distance). 
We have shown the lack of self-similar collapses in terms of analytic functions of intrinsic distance on the hyperbolic plane. 
But one may consider embedding $\mathbb H^2$ in $\mathbb R^3$, and seeking extrinsically self-similar collapse. By embedding into Euclidean three-space, the equations work only locally, but this is sufficient for considering collapse. Given two points on the hyperbolic plane, the Euclidean chord length is then a function not only of the intrinsic geodesic distance, but also on the specific location of the points. It might be possible then that one may find a specific collapsing configuration which does so self-similarly in terms of the extrinsic distance, yet this need not be a generic feature of the collapse.  Perhaps a more natural embedding of $\mathbb H^2$ is into Minkowski space $\mathbb R^{2, 1}$ (e.g. $\mathbb{R}^3$ endowed with the Minokowski metric). Now, the extrinsic distance is Lorentzian. Indeed, this ``matching" of geometries resolves the dependence of chord length on position, is related to the intrinsic geodesic distance as $2\sinh(d_{ij}/2)$, and becomes a more direct analog of the embedding of $\mathbb S^2$ in $\mathbb R^3$. However, this then falls under the assumptions of Theorem \ref{thm:1-1}, and as such no collapses in terms of this extrinsic distance can occur.

All this said, we leave open the interesting possibility that, upon discovery of the ``right notion" of inter-vortex distance, the collapse could be self-similar.
\end{remark}

\section{Proof of existence}
In the previous section, we considered specifically the motion that is self-similar. However, as mentioned previously, one would not generally expect collapse to be self-similar. As such, we wish to prove the existence of collapsing three-vortex configurations on a general surface $S$. To do this, we shall show that any planar self-similar three vortex collapse can be (approximately) realized in a sufficiently small geodesic neighborhood of a point $p\in S$, and that for a sufficiently small initial configuration, the surface dynamics remain close to the planar dynamics up to a rescaling in time. We show that for sufficiently small neighborhoods, the ``non-planar" terms are bounded, and therefore by the results of \cite{grotto2022burst}, the dynamics remains regular and we have collapsing configurations that are approximately planar. In principle, any configuration can then be found by using time-reversal symmetry and evolving infinitesimal configurations in time.

Furthermore, as the configuration shrinks, the non-planar corrections to the equations of motion become smaller. As such, as the configuration approaches collapse, the motion tends to become self-similar, and so the collapse is asymptotically self-similarly. 

We work with point vortices of nonzero circulation on a surface $S$ embedded in $\mathbb R^3$. Throughout, we shall assume that $S$ is oriented. Let $g$ denote the metric induced by the embedding of $S$. The tangent bundle $TS$ is endowed with the metric $g$ and the corresponding almost complex structure (rotation by $\pi/2$) $J(p) : T_pS \mapsto T_pS$ that rotates the tangent plane by 90 degrees. 

Following \cite{grotto2022burst}, introduce the space of functions
\begin{equation*}
    E_T = C([0, T), C^2(\mathbb {R}^2, \mathbb {R}^2)) \cap \mathrm{Lip}([0, T), C(\mathbb {R}^2,\mathbb {R}^2))
\end{equation*}
for time $T>0$ endowed with the norm $\lVert{f}\rVert_{E_T} = \sup_{t \in [0, T)} \lVert{f_t}\rVert_{C^2} + [f]_\mathrm{Lip}$. Here, $\lVert{f}\rVert_{C^k}$ is the sum of supremum norms of the first $k$-derivatives of $f$ for some $f:\mathbb R^2 \mapsto \mathbb R^2$ with continuous $k$-th derivative, and $[f]_\mathrm{Lip}$ denotes the Lipschitz seminorm. It is then proven that, given that a perturbation is sufficiently well-behaved with respect to the $E_T$ norm, the self-similar motion of point vortices persists under said perturbation, as stated in Lemma \ref{thm:prop-2.3-from-paper}.
\begin{lem}\label{thm:prop-2.3-from-paper}
Let $M,\rho>0$. Then there exists $T^*>0$ such that for every $T<T^*$ and every
$f\in E_T$ with $\|f\|_{E_T}\le M$, the system
\begin{equation*}
\dot z_j(t)=\frac{1}{2\pi}\sum_{k\ne j}\Gamma_k
\frac{(z_j(t)-z_k(t))^\perp}{|z_j(t)-z_k(t)|^2}+f(t,z_j(t)),
\qquad j=1,2,3,
\end{equation*}
admits a $C^1$ solution $z:(0,T)\to\mathbb R^6\setminus\Delta^3$
satisfying
\begin{equation*}
    \lim_{t\to0}z_j(t)=0,\qquad
    [z_j]_{C_T^{1/2}}\le 3[w_j]_{C^{1/2}},\quad j=1,2,3,
\end{equation*}
and
\begin{equation*}
    \sup_{0<t<T}|z(t)|\le \rho .
\end{equation*}
\end{lem}

To apply this result, we must first control the behavior of the Green's function. Following \cite{drivas2024vortexpairs}, the singular behavior of the Green's function may be isolated. 
\begin{lem}\label{lem:green-func}
The Green's function of the surface Laplacian may be written as 
\begin{equation*}
    G_S(p, q) = -\frac{1}{2\pi}\log d_S(p, q) + R_S(p, q).
\end{equation*}
where, for all $\alpha\in(0, 1)$, there is a constant $C= C(S, \alpha)$ such that the ``regular part" satisfies
\begin{equation*}
    \lVert{R_S(\cdot, q)}\rVert_{C^{2, \alpha}(S)}\leq C\quad for ~all \quad q\in S.
\end{equation*}
\end{lem}
\noindent
Using this decomposition, we write the Hamiltonian for any finite number of  vortices as 
\begin{equation}
    H(z) = -\frac{1}{2\pi}\sum_{i<j} \Gamma_i \Gamma_j  \log \lVert{z_i - z_j}\rVert + \mathsf{Reg}_S(z),
\end{equation}
where 
\begin{equation}
    \mathsf{Reg}_S(z) := - \frac{1}{2\pi}\sum_{i<j} \Gamma_i \Gamma_j \log \frac{d_S(z_i, z_j)}{\lVert{z_i-z_j}\rVert} + \sum_{i<j} \Gamma_i \Gamma_j R_S(z_i, z_j)  + \frac{1}{2}\sum_\ell \Gamma_\ell^2 R_S(z_\ell, z_\ell).
\end{equation}

Now, we aim to switch to local coordinates for some point on the surface, with respect to which we shall do our analysis. 
\begin{lem}
Let S be a surface embedded in $\mathbb{R}^3$, and fix a point $p \in S$. Let $\exp_p : T_pS \supset B_{R_0}(0)\mapsto S$ be the exponential map at let $\varphi = \exp_p^{-1}$ be the normal coordinates on $U := \exp_p(B_{R_0}(0))$. We write $\zeta_i = \varphi(z_i)$. Then, there exists constants $r_1 >0$ and $C>0$ depending on $S, p$, and the circulations, such that the following holds:
For every configuration in the neighborhood, the equations of motion may be written as 
\begin{equation}
    \dot \zeta_i = -\frac{1}{2\pi} J(0)\sum_{j \neq i} \Gamma_j \frac{\zeta_i - \zeta_j}{|\zeta_i - \zeta_j|^2} + \mathcal E_i(\zeta),
\end{equation}
where the remainder satisfies 
\begin{equation*}
    |\mathcal E_i| \leq C \left(|\zeta_i| \sum_{j \neq i} \frac{1}{|\zeta_i - \zeta_j|} + \sum_{j \neq i} (|\zeta_i| + |\zeta_j| + |\zeta_i - \zeta_j|) + 1 \right).
\end{equation*}
In particular, for some fixed $R,\kappa>0$, and some
$\varepsilon\in(0,r_1/R]$, suppose
\begin{equation*}
    |\zeta_i|\leq R\varepsilon,
    \qquad
    |\zeta_i-\zeta_j|\geq \kappa s(\zeta),
    \qquad
    s(\zeta):=\max_{1\leq \ell\leq 3}|\zeta_\ell|.
\end{equation*}
Then
\begin{equation*}
    |\mathcal E_i(\zeta)|\leq C_{R,\kappa}
\end{equation*}
for all $i$.
\end{lem}
\begin{proof}
We work in normal coordinates centered at $p$. In these coordinates, for a sufficiently small neighborhood, the metric satisfies
\begin{equation*}
    g_{ij}(\zeta)=\delta_{ij}+Q_{ij}(\zeta), \qquad |Q_{ij}(\zeta)|\leq C_g|\zeta|^2,
\end{equation*}
and the almost complex structure satisfies
\begin{equation*}
    J(\zeta)=J(0)+R_J(\zeta), \qquad |R_J(\zeta)|\leq C_J|\zeta|.
\end{equation*}
We write the Hamiltonian in the form
\begin{equation*}
    H(\zeta)=-\frac{1}{2\pi}\sum_{i<j}\Gamma_i\Gamma_j\log|\zeta_i-\zeta_j|+\Reg_S(\zeta)+\mathsf{Err}(\zeta),
\end{equation*}
where
\begin{equation*}
    \mathsf{Err}(\zeta)=-\frac{1}{2\pi}\sum_{i<j}\Gamma_i\Gamma_j\log\left(\frac{\|\exp_p(\zeta_i)-\exp_p(\zeta_j)\|}{|\zeta_i-\zeta_j|}\right).
\end{equation*}

We first bound the contribution from $\Reg_S$. By definition, $\Reg_S$ contains the Robin terms together with the correction
\begin{equation*}
    -\frac{1}{2\pi}\sum_{i<j}\Gamma_i\Gamma_j\log\left(\frac{d_S(z_i,z_j)}{\|z_i-z_j\|}\right).
\end{equation*}
In normal coordinates, the geodesic distance and the ambient Euclidean distance satisfy
\begin{equation*}
    d_S(\exp_p\zeta_i,\exp_p\zeta_j)=|\zeta_i-\zeta_j|\left(1+D^g_{ij}(\zeta_i,\zeta_j)\right),
\end{equation*}
and
\begin{equation*}
    \|\exp_p(\zeta_i)-\exp_p(\zeta_j)\|=|\zeta_i-\zeta_j|\left(1+D^e_{ij}(\zeta_i,\zeta_j)\right).
\end{equation*}
After shrinking the coordinate neighborhood if necessary, these errors satisfy
\begin{equation*}
    |D^g_{ij}|+|D^e_{ij}|\leq C\left((|\zeta_i|+|\zeta_j|)^2+|\zeta_i-\zeta_j|^2\right),
\end{equation*}
and
\begin{equation*}
    |\nabla_{\zeta_i}D^g_{ij}|+|\nabla_{\zeta_i}D^e_{ij}|\leq C\left(|\zeta_i|+|\zeta_j|+|\zeta_i-\zeta_j|\right).
\end{equation*}
Therefore,
\begin{equation*}
    \log\left(\frac{d_S(\exp_p\zeta_i,\exp_p\zeta_j)}{\|\exp_p(\zeta_i)-\exp_p(\zeta_j)\|}\right)=\log\left(\frac{1+D^g_{ij}}{1+D^e_{ij}}\right),
\end{equation*}
and hence
\begin{equation*}
    \left|\nabla_{\zeta_i}\log\left(\frac{d_S(\exp_p\zeta_i,\exp_p\zeta_j)}{\|\exp_p(\zeta_i)-\exp_p(\zeta_j)\|}\right)\right|\leq C\left(|\zeta_i|+|\zeta_j|+|\zeta_i-\zeta_j|\right).
\end{equation*}
Together with the regularity of the Robin terms, this gives
\begin{equation*}
    |\nabla_{\zeta_i}\Reg_S(\zeta)|\leq C.
\end{equation*}
Next, we estimate $\mathsf{Err}$. By the expansion of the ambient Euclidean distance in normal coordinates,
\begin{equation*}
    \|\exp_p(\zeta_i)-\exp_p(\zeta_j)\|=|\zeta_i-\zeta_j|\left(1+D^e_{ij}(\zeta_i,\zeta_j)\right),
\end{equation*}
with $D^e_{ij}$ satisfying the bounds above. Hence
\begin{equation*}
    \log\left(\frac{\|\exp_p(\zeta_i)-\exp_p(\zeta_j)\|}{|\zeta_i-\zeta_j|}\right)=\log(1+D^e_{ij}).
\end{equation*}
Since $D^e_{ij}$ is small on $U$, we obtain
\begin{equation*}
    \left|\nabla_{\zeta_i}\log\left(\frac{\|\exp_p(\zeta_i)-\exp_p(\zeta_j)\|}{|\zeta_i-\zeta_j|}\right)\right|\leq C\left(|\zeta_i|+|\zeta_j|+|\zeta_i-\zeta_j|\right).
\end{equation*}
Therefore,
\begin{equation*}
    |\nabla_{\zeta_i}\mathsf{Err}(\zeta)|\leq C\sum_{j\neq i}\left(|\zeta_i|+|\zeta_j|+|\zeta_i-\zeta_j|\right).
\end{equation*}
Let us denote the flat Hamiltonian by
\begin{equation*}
    H_0:=-\frac{1}{2\pi}\sum_{i<j}\Gamma_i\Gamma_j\log|\zeta_i-\zeta_j|.
\end{equation*}
Thus $H=H_0+\Reg_S+\mathsf{Err}$. We have the bound on the surface derivative
\begin{equation*}
    \widetilde\nabla_{\zeta_i}H=\nabla_{\zeta_i}H+\mathcal G_i(\zeta),
\end{equation*}
with
\begin{equation*}
    |\mathcal G_i(\zeta)|\leq C|\zeta_i|^2|\nabla_{\zeta_i}H(\zeta)|.
\end{equation*}
Using
\begin{equation*}
    |\nabla_{\zeta_i}H_0|\leq C\sum_{j\neq i}\frac{1}{|\zeta_i-\zeta_j|},
\end{equation*}
together with the bounds on $\nabla_{\zeta_i}\Reg_S$ and $\nabla_{\zeta_i}\mathsf{Err}$, we get
\begin{equation*}
    |\mathcal G_i(\zeta)|\leq C|\zeta_i|^2\left(\sum_{j\neq i}\frac{1}{|\zeta_i-\zeta_j|}+\sum_{j\neq i}\left(|\zeta_i|+|\zeta_j|+|\zeta_i-\zeta_j|\right)+1\right).
\end{equation*}
Since we may assume $|\zeta_i|\leq r_1\leq 1$ by shrinking $r_1$ if necessary, this implies
\begin{equation*}
    |\mathcal G_i(\zeta)|\leq C\left(|\zeta_i|\sum_{j\neq i}\frac{1}{|\zeta_i-\zeta_j|}+1\right),
\end{equation*}
where the remaining lower-order terms have been absorbed into the constant.

Using $J(\zeta_i)=J(0)+R_J(\zeta_i)$, we write
\begin{equation*}
    \Gamma_i\dot{\zeta}_i=J(0)\nabla_{\zeta_i}H_0+R_J(\zeta_i)\nabla_{\zeta_i}H_0+J(\zeta_i)\left(\nabla_{\zeta_i}\Reg_S+\nabla_{\zeta_i}\mathsf{Err}+\mathcal G_i\right).
\end{equation*}
The leading term is
\begin{equation*}
    J(0)\nabla_{\zeta_i}H_0=-\frac{\Gamma_i}{2\pi}J(0)\sum_{j\neq i}\Gamma_j\frac{\zeta_i-\zeta_j}{|\zeta_i-\zeta_j|^2}.
\end{equation*}
For the second term, we use the bound on $R_J$ together with the estimate for $\nabla H_0$ to obtain
\begin{equation*}
    |R_J(\zeta_i)\nabla_{\zeta_i}H_0|\leq C|\zeta_i|\sum_{j\neq i}\frac{1}{|\zeta_i-\zeta_j|}.
\end{equation*}
Since $J(\zeta_i)$ is bounded on $U$, the remaining terms satisfy
\begin{equation*}
    \left|J(\zeta_i)\left(\nabla_{\zeta_i}\Reg_S+\nabla_{\zeta_i}\mathsf{Err}+\mathcal G_i\right)\right|\leq C\left(|\zeta_i|\sum_{j\neq i}\frac{1}{|\zeta_i-\zeta_j|}+\sum_{j\neq i}\left(|\zeta_i|+|\zeta_j|+|\zeta_i-\zeta_j|\right)+1\right).
\end{equation*}
Combining these estimates and dividing by $|\Gamma_i|$, we obtain
\begin{equation}\label{eq:local}
    \dot{\zeta}_i=-\frac{1}{2\pi}J(0)\sum_{j\neq i}\Gamma_j\frac{\zeta_i-\zeta_j}{|\zeta_i-\zeta_j|^2}+\mathcal E_i(\zeta),
\end{equation}
where
\begin{equation*}
    |\mathcal E_i(\zeta)|\leq C\left(|\zeta_i|\sum_{j\neq i}\frac{1}{|\zeta_i-\zeta_j|}+\sum_{j\neq i}\left(|\zeta_i|+|\zeta_j|+|\zeta_i-\zeta_j|\right)+1\right).
\end{equation*}
Finally, suppose that
\begin{equation*}
    |\zeta_i|\leq R\varepsilon, \qquad |\zeta_i-\zeta_j|\geq \kappa s(\zeta), \qquad s(\zeta):=\max_{1\leq \ell\leq 3}|\zeta_\ell|.
\end{equation*}
Then
\begin{equation*}
    |\zeta_i|\sum_{j\neq i}\frac{1}{|\zeta_i-\zeta_j|}\leq s(\zeta)\sum_{j\neq i}\frac{1}{\kappa s(\zeta)}\leq \frac{2}{\kappa}.
\end{equation*}
Moreover,
\begin{equation*}
    \sum_{j\neq i}\left(|\zeta_i|+|\zeta_j|+|\zeta_i-\zeta_j|\right)\leq C_R\varepsilon\leq C_R.
\end{equation*}
Hence
\begin{equation*}
    |\mathcal E_i(\zeta)|\leq C_{R,\kappa}
\end{equation*}
for all $i$, completing the proof.
\end{proof}

We now complete the main proof of Theorem \ref{prop:main}.
\begin{proof}
Introduce the scaled variables 
\begin{equation*}
    \zeta_i(t) = \varepsilon \xi_i(\tau), \qquad \tau=\frac{t}{\varepsilon^2}.
\end{equation*}
Then, $
    \dot \zeta_i(t) = \frac{1}{\varepsilon} \frac{d\xi_i}{d\tau}.$
Substituting this into the local equation (\ref{eq:local}) from the lemma yields
\begin{equation}\label{eq:rescaled}
    \frac{d\xi_i}{d\tau} = -\frac{1}{2\pi}J(0) \sum_{j\neq i}\Gamma_j \frac{\xi_i-\xi_j}{|\xi_i-\xi_j|^2} + f_i(\xi).
\end{equation}
where $
    f_i(\xi) := \varepsilon \mathcal E_i(\varepsilon \xi).$
Hence, the rescaled dynamics is now a planar three-vortex system plus a perturbation $f_i$. In rescaled coordinates, we use the scale-invariant nondegeneracy region
\begin{equation*}
    \mathcal C_{\kappa} := \left\{ \xi\in(\mathbb R^2)^3\setminus\Delta^3: 0<r(\xi)\leq 1,\quad |\xi_i-\xi_j|\geq \kappa r(\xi)
\text{ for all }i\neq j \right\},
\end{equation*}
where
\begin{equation*}
    r(\xi):=\max_{1\leq \ell\leq 3}|\xi_\ell|.
\end{equation*}
If $\xi\in\mathcal C_{\kappa}$, then for
\(\zeta_i=\varepsilon\xi_i\) we have
\begin{equation*}
    |\zeta_i|\leq \varepsilon, \qquad |\zeta_i-\zeta_j| = \varepsilon|\xi_i-\xi_j| \geq \varepsilon\kappa r(\xi) = \kappa s(\zeta).
\end{equation*}
Therefore, by the previous lemma, applied with unit radius,
\begin{equation*}
    |\mathcal E_i(\varepsilon\xi)|\leq C_{\kappa}.
\end{equation*}
Hence
\begin{equation*}
    |f_i(\xi)|\leq C_{\kappa}\varepsilon.
\end{equation*}
Since the planar collapsing profile $\xi^*$ is self-similar and nondegenerate, we may normalize it so that \(r(\xi^*(\tau))\leq 1\) on the time interval under consideration. The ratios
\begin{equation*}
    \frac{|\xi_i^*(\tau)-\xi_j^*(\tau)|}{r(\xi^*(\tau))}
\end{equation*}
are constant in $\tau$ and strictly positive. Thus, after choosing
$\kappa>0$ appropriately, we have
\begin{equation*}
    \xi^*(\tau)\in\mathcal C_{\kappa}
\end{equation*}
for $0<\tau<T$. By taking the perturbed solution sufficiently close to
$\xi^*$, and replacing $\kappa$ by a smaller constant if necessary, we may assume that
\begin{equation*}
    \xi(\tau)\in\mathcal C_{\kappa}
\end{equation*}
for $0<\tau<T$. Hence, the bounds on $\mathcal E_i$ remain valid along the trajectory.

Since $\lVert{f_i}\rVert_{E_T} \to 0$ as $\varepsilon \to 0$, there exists $\varepsilon_0$ such that for every $0 < \varepsilon < \varepsilon_0$, the rescaled perturbed system admits a solution $\xi(\tau)$ remaining close to $\xi^*(\tau)$, corresponding to a collapsing trajectory. Returning to the original surface variables by 
\begin{equation*}
    \zeta_i(t) = \varepsilon \xi_i(\tau), \qquad \tau=\frac{t}{\varepsilon^2},
\end{equation*}
it follows that 
\begin{equation*}
    \zeta_i(t) \to 0 \quad \mathrm{as} ~ t  \to t^*,
\end{equation*}
and hence 
\begin{equation*}
    z_i(t) \to p.
\end{equation*}
Thus, the three vortices collapse to $p$.

It remains to verify that the collapse is asymptotically self-similar. Starting with the rescaled equation of motion (\ref{eq:rescaled}), let us define $q_{ij} = \xi_i - \xi_j$. Then, 
\begin{equation}
    \frac{d}{d\tau} q_{ij}^2 = \frac{2\Gamma_k A}{\pi}\left(\frac{1}{q_{jk}^2} - \frac{1}{q_{ik}^2} \right) + 2 q_{ij}\cdot(f_i - f_j),
\end{equation}
where $A$ is the signed area of the triangle, $A = \frac{1}{2}q_{ij} \cdot J(0) r_A = \frac{1}{2}q_{ij} \cdot J(0) q_{ik}$. The first term is precisely the same as on the plane, whereas the second term is the error from the surface. Using the same bounds as before, $|2 q_{ij} \cdot (f_i - f_j)|\leq 4 C_{\kappa} \varepsilon|q_{ij}|$. Hence, as $\tau \to \tau_c$, the error term vanishes. Equivalently, the conditions in Definition \ref{def.near-self-similar} is satisfied, hence collapse is asymptotically self-similar. Furthermore, due to the local Euclidean nature of $S$, every collapse of three vortices on $S$ must be asymptotically self-similar.
\end{proof}

\section*{Acknowledgements}
We would like to thank Nathan Carlson, Daniil Glukhovskiy, and Boris Khesin for helpful discussions.  In particular, we would like to thank B. Khesin for informing us of \cite{simo2022n}.
The work of TDD was supported by the NSF CAREER award \#2235395, a Stony Brook
University Trustee’s award as well as an Alfred P. Sloan Fellowship. TDD acknowledges support of the Institut Henri Poincaré (UAR 839 CNRS-Sorbonne Université), and LabEx CARMIN (ANR-10-LABX-59-01).

\bibliographystyle{unsrt}  
\bibliography{bibliography}

\end{document}